\documentclass[research]{fcs}
\usepackage{layout}
\usepackage{braket}
\usepackage{color}
\usepackage{ulem}

\newtheorem{remark}{Remark}

\title{Revisiting fixed-point quantum search: proof of the quasi-Chebyshev lemma}
\author[1]{Guanzhong~Li}
\author[1]{Shiguang~Feng}
\author[1,2,+]{Lvzhou~Li}
\address[1]{Institute of Quantum Computing and Software, School of Computer Science and Engineering, Sun Yat-sen University, Guangzhou 510006, China}
\address[2]{Quantum Science Center of Guangdong-Hong Kong-Macao Greater Bay Area (Guangdong), Shenzhen 518045, China}
\corremail{lilvzh@mail.sysu.edu.cn}
\fcssetup{
  received = {month dd, yyyy},
  accepted = {month dd, yyyy},
}

\begin{abstract}
The original Grover's algorithm suffers from the souffle problem, which means that the success probability of quantum search decreases dramatically if the iteration time is too small or too large from the right time.
To overcome the souffle problem, the fixed-point quantum search with an optimal number of queries was proposed [Phys. Rev. Lett. 113, 210501 (2014)], which always finds a marked state with a high probability  when a lower bound of the proportion of marked states is given.
The fixed-point quantum search relies on a key lemma regarding the explicit formula of recursive quasi-Chebyshev polynomials, but its proof is not given explicitly.
In this work, we give a detailed proof of this lemma, thus providing a sound foundation for the correctness of the fixed-point quantum search.
This lemma may be of independent interest as well, since it expands the mathematical form of the recursive relation of Chebyshev polynomials of the first kind, and it also constitutes a key component in overcoming the souffle problem of quantum walk-based search algorithms, for example, robust quantum walk search on complete bipartite graphs [Phys. Rev. A 106, 052207 (2022)].
The lemma is also central to a recently proposed quantum algorithm named quantum phase discrimination, which has become a fundamental subroutine in quantum search on graphs [arxiv: 2504.15194].
Hopefully, more applications of the lemma will be found in the future.
\end{abstract}
\keywords{Quantum algorithm, Grover's algorithm, fixed-point quantum search, quasi-Chebyshev polynomials}

\begin{document}
\section{Introduction}
Grover’s search~\cite{Grover} is a fundamental algorithm in quantum computing, as it achieves a quadratic speedup for the unstructured search problem.
It can be easily generalized to amplitude amplification~\cite{QAA}, which has become a ubiquitous component in designing quantum algorithms.
Stimulated by its wide application, the research on quantum search itself is never-ending and has led to a variety of versions, such as the deterministic version in which a marked state is obtained with certainty when the proportion is known beforehand~\cite{QAA,arbi_phase,Long,exact}, the variable time version where the cost of each oracle query is different~\cite{Ambainis2010,ambainis_et_al2023}, and versions handling oracles with bounded-error~\cite{bounded2003,rosmanis2023quantum}.

In this paper, we revisit versions that deal with the souffle problem, i.e. the problem that the success probability of quantum search decreases dramatically if the iteration time is too small or too large from the right time.
Brassard ~\cite{souffle} pointed out this phenomenon in a figurative way:
``Quantum searching is like cooking a souffle. You put the state obtained by quantum parallelism in a `quantum oven' and let the desired answer rise slowly.
Success is almost guaranteed if you open the oven at just the right time.
But the souffle is very likely to fall—the amplitude of the correct answer drops to zero—if you open the oven too early.
Furthermore, the souffle could burn if you overcook it; strangely, the amplitude of the desired state starts shrinking after reaching its maximum.''

A straightforward solution to the souffle problem is first to estimate the proportion of the marked states using amplitude estimation~\cite{QAA},
and then perform a quantum search based on that estimation.
An alternative way is to repeat the quantum search with exponentially increasing random iteration time, until a marked state is obtained~\cite{tight1998}.
Grover himself has also studied the souffle problem and proposed the $\pi/3$ method~\cite{fixed_point2005} with a `fixed-point' property, such that the success probability monotonically increases with the iteration times.
However, the quadratic speedup is lost in this method.
In 2014, Yoder, Low, and Chuang~\cite{fixed_point2014} proposed an innovative quantum search algorithm that overcomes the souffle problem while maintaining the quadratic speedup.
The proposed algorithm achieves a slightly different `fixed-point' property:
as long as the actual proportion of marked states $\lambda^2$ is greater than a predetermined lower bound $w^2$,
a marked state is guaranteed to be found with probability greater than $1-\delta^2$,
using approximately $\ln(2/\delta)/w$ oracle queries.

To achieve the goal above, the generalized Grover's iteration $G(\alpha,\beta)$ (see Eq.~\eqref{eq:generalized_iter} below for its definition) is used, and the sequence of angles $(\alpha_k,\beta_k)_{k=1}^{l}$ is given explicitly~\cite{fixed_point2014}.
The correctness of the algorithm, i.e. the closed-form angle parameters $(\alpha_k,\beta_k)_{k=1}^{l}$, relies on the explicit formula of a recursive quasi-Chebyshev polynomial.
More specifically, it was stated in Ref.~\cite{fixed_point2014} that ``it can be shown using combinatorial arguments analogous to those in~\cite{counting} that $a_L^{(\gamma)}(x) = T_L(x/\gamma) / T_L(1/\gamma)$'', where $a_L^{(\gamma)}(x)$ is a complex, degree-$L$ polynomial that generalizes the Chebyshev polynomials.
However, it seems unclear how to prove this key formula, since Ref.~\cite{counting} only provides the combinatorial interpretation of $T_L(x)$, i.e. Chebyshev polynomials of the first kind, and the paper~\cite{fixed_point2014} does not provide any further explanation.

Some years later, an alternative and rigorous method (see \cite[Theorem 27]{qsvt} or \cite[Section III]{Unification}) was proposed to achieve the same goal using quantum singular value transformation (QSVT), a framework that unifies many quantum algorithms such as amplitude amplification, Hamiltonian simulation, eigenvalue filtering, and quantum linear system problems.
Roughly speaking, this method consists of two steps.
The first step concerns constructing a polynomial that approximates the sign function, and the second step concerns obtaining the sequence of angles from the coefficient of the polynomial.
However, since both steps involve much numerical approximation, it's unlikely to obtain closed-form angle parameters as in Ref.~\cite{fixed_point2014}.
A challenging task of the QSVT framework is the computation of the angle parameters with efficiency and numerical stability, and there has been a series of works on optimizing this procedure~\cite{Haah2019product,angles_MP,phase_factor}.
This also highlights the importance of the fixed-point quantum search in Ref.~\cite{fixed_point2014}, as it appears to be the only nontrivial instance of QSVT that provides closed-form angle parameters.

In this work, we give a detailed proof of the explicit formula of a recursive quasi-Chebyshev polynomial, i.e. $a_L^{(\gamma)}(x) = T_L(x/\gamma) / T_L(1/\gamma)$,
which is fundamental to the closed-form angle parameters of the fixed-point quantum search~\cite{fixed_point2014}.
We restate this key formula in Lemma~\ref{lem:chebyshev}, which expands the mathematical form of the original recursive relation of Chebyshev polynomials, since $a_L^{(\gamma)}(x)$ reduces to $T_L(x)$ when $\gamma=1$.
For completeness, we also provide a proof of the correctness of the fixed-point quantum search using Lemma~\ref{lem:chebyshev}, by relating the algorithm's failure amplitude to quasi-Chebyshev polynomials.

The quasi-Chebyshev lemma is a key component in the {\it robust} quantum walk search algorithm on complete bipartite graphs proposed in Ref.~\cite{robustQW}.
The proposed quantum walk search algorithm is robust in the sense that it can find a marked vertex on an $N$-vertex complete bipartite graph with probability greater than $1-\epsilon$ as long as the number of quantum walk search steps $h \geq \ln(\frac{2}{\sqrt{\epsilon}})\sqrt{N}+1$ for any adjustable parameter $\epsilon$, without knowing the number of marked vertices or sacrificing the quadratic quantum speedup.

The quasi-Chebyshev lemma is also central to a recently proposed quantum algorithm named quantum phase discrimination (QPD)~\cite{QPD}, which solves the problem of deciding whether the eigenphase $\theta\in(-\pi,\pi]$ of a given eigenstate $\ket{\psi}$ with eigenvalue $e^{i\theta}$ is zero or not, using applications of the unitary $U$ provided as a black box.
QPD achieves optimal query complexity and the quantum circuit is simple, consisting of only one ancillary qubit and a sequence of controlled-$U$ interleaved with single qubit $Y$ rotations, whose angles are given by a simple analytical formula that stems from the quasi-Chebyshev lemma.
QPD could also become a fundamental subroutine in other quantum algorithms, as two applications to quantum search on graphs are presented, one to spatial search on graphs and the other to path-finding on graphs~\cite{QPD}.
More specifically, a new quantum approach to the spatial search problem is provided based on QPD, which can find a marked vertex with probability $\Omega(1)$ in total evolution time $ O(\frac{1}{\eta \sqrt{\varepsilon}})$ and query complexity $ O(\frac{1}{\sqrt{\varepsilon}})$, where $\eta$ is the gap between the zero and non-zero eigenvalues of the graph Laplacian and $\varepsilon$ is a lower bound on the proportion of marked vertices.
The query complexity of a path-finding quantum algorithm on a welded-tree circuit graph with $\Theta(n2^n)$ vertices~\cite{multi_electric} can also be reduced from $\tilde{O}(n^{11})$ to $\tilde{O}(n^8)$ using QPD.

The rest of this paper is organized as follows.
In Section~\ref{sec:review}, we review the fixed-point quantum search that overcomes the souffle problem, which is summarized in Theorem~\ref{thm:main}.
In Section~\ref{sec:proof_thm}, we restate the key formula of recursive quasi-Chebyshev polynomials in Lemma~\ref{lem:chebyshev}, and then prove Theorem~\ref{thm:main} using Lemma~\ref{lem:chebyshev} and their relation (Lemma~\ref{lem:phase_amp}) for completeness.
Section~\ref{sec:prooflemma1} is devoted to the proof of Lemma~\ref{lem:chebyshev}, which is the main contribution of this article.

\section{Fixed-point quantum search}\label{sec:review}
We begin by defining the generalized Grover's iteration $G(\alpha,\beta)$ which will be used later in the fixed-point quantum search.
\begin{equation}\label{eq:generalized_iter}
    G(\alpha,\beta) := S_0(\beta) S_M(\alpha),
\end{equation}
where $S_0(\beta) := e^{i\beta \ket{\psi_0}\bra{\psi_0}} = I - (1-e^{i\beta})\ket{\psi_0}\bra{\psi_0}$ multiplies a phase shift of $e^{i\beta}$ to the initial state $\ket{\psi_0}$;
and $S_M(\alpha) := e^{i\alpha \Pi_M } = I - (1-e^{i\alpha})\sum_{m\in M} \ket{m}\bra{m}$ multiplies a phase shift of $e^{i\alpha}$ to all the marked basis states and leaves the other states unchanged.
The phase oracle $S_M(\alpha)$ can be implemented using twice the standard oracle $O: \ket{x}\ket{b} \mapsto \ket{x}\ket{b\oplus \delta(x\in M)}$ that flips the auxiliary qubit $\ket{b}$ when the basis state $x$ is marked~\cite{fixed_point2014}.


Let's have a quick review of Grover's original search algorithm, which uses the restricted iteration $G(\pi,\pi)$ with $\alpha=\beta=\pi$ in Eq.~\eqref{eq:generalized_iter}.
The initial state $\ket{\psi_0}$ is the equal-superposition of all basis states.
Each iteration $-G(\pi,\pi) =(2\ket{\psi_0}\bra{\psi_0}-I)  (I-2\sum_{m\in M} \ket{m}\bra{m}) $ can be seen as the composition of two reflections, and is thus a rotation.
The more detailed geometric interpretation is shown in Fig.~\ref{fig:2d_original}.

\begin{figure}[htbp]
    \centering
    \includegraphics[width=0.4\linewidth]{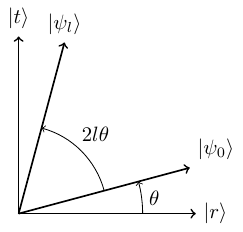}
    \caption{Geometric interpretation of Grover's search process, i.e. $\ket{\psi_l} := [-G(\pi,\pi)]^l \ket{\psi_0}$,
    in the invariant subspace $\mathcal{H}_0 = \mathrm{span} \{\ket{r}, \ket{t}\}$,
    where $\ket{t} := \Pi_{M}\ket{\psi_0}/\lambda$, $\lambda := \left\| \Pi_M \ket{\psi_0} \right\|$
    and $\ket{r} := \Pi_{M}^\bot \ket{\psi_0}/\sqrt{1-\lambda^2}$, $\Pi_{M}^\bot := I -\Pi_{M}$.
    One iteration $-G(\pi,\pi) = (2\ket{\psi_0}\bra{\psi_0}-I)  (I-2\ket{t}\bra{t})$ is interpreted as a reflection around $\ket{r}$ followed by a reflection around $\ket{\psi_0}$, and is thus a counter-clockwise rotation by $2\theta := 2\arcsin\lambda$. }
    \label{fig:2d_original}
\end{figure}

The souffle problem has a nice geometric explanation as shown in Fig.~\ref{fig:2d_original}.
The success probability (i.e. projection of $\ket{\psi_l}$ onto $\ket{t}$) decreases dramatically if the iteration time $l$ is too far from the optimal time $\lfloor (\frac{\pi}{2}-\theta)/(2\theta) \rceil = \lfloor \frac{\pi}{4\arcsin\lambda} - \frac{1}{2} \rceil$ which depends on $\lambda = \left\| \Pi_M \ket{\psi_0} \right\|$.

To solve the souffle problem, we can use the fixed-point quantum search in Ref.~\cite{fixed_point2014} such that only a lower bound $w$ of $\lambda = \left\| \Pi_M \ket{\psi_0} \right\|$ is required to be known beforehand.

\begin{theorem}[\cite{fixed_point2014}]\label{thm:main}
    For any $w\in (0,1)$ and $\delta \in (0,1)$,
    consider the following procedure:
\begin{equation}\label{eq:procedure}
\ket{\psi_l} := G(\alpha_l,\beta_l) \cdots G(\alpha_1,\beta_1) \ket{\psi_0},
\end{equation}
    where the iteration time $l \geq \ln(2/\delta)/(2w)$,
    and the sequence of parameters are set according to 
    \begin{align*}
        \alpha_k &= 2\mathrm{arccot}\left( w \tan(\frac{2k-1}{2l+1}\pi) \right), \\
        \beta_k &= -2\mathrm{arccot}\left( w \tan(\frac{2k}{2l+1}\pi) \right),
    \end{align*}
    for $k=1\sim l$.
    Then $\left\| \Pi_M \ket{\psi_l} \right\| \geq \sqrt{1-\delta^2}$
    as long as $\left\| \Pi_M \ket{\psi_0} \right\| \geq w$.
\end{theorem}

To have an intuitive understanding of the `fixed-point' property in Theorem~\ref{thm:main}, the plot of how $\left\|\Pi_M|\psi_l\rangle\right\|$ varies with $\left\|\Pi_M|\psi_0\rangle\right\|$ is shown in Fig~\ref{fig:thm1_examp}.

\begin{figure}[htbp]
    \centering
    \includegraphics[width=0.7\linewidth]{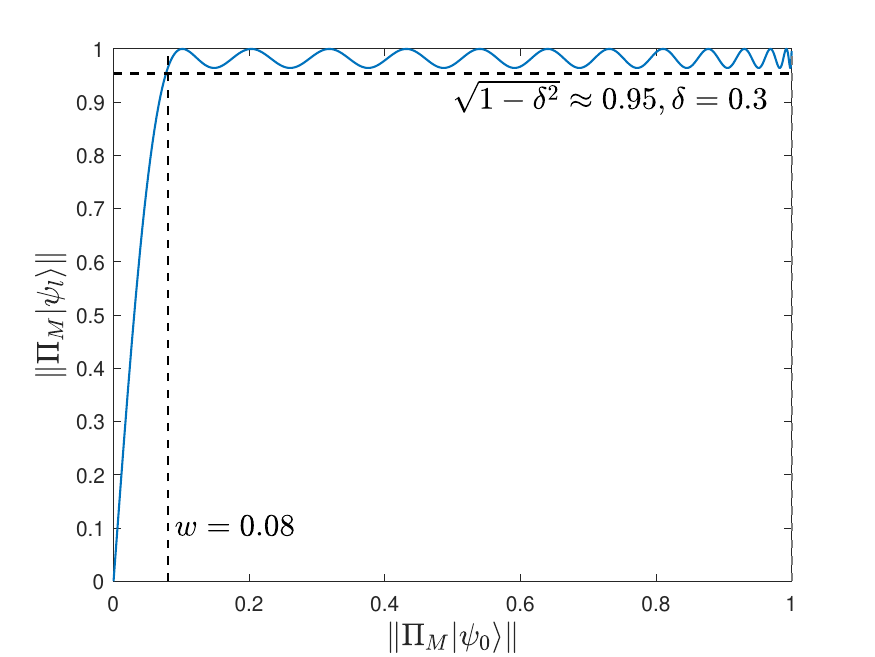}
    \caption{The plot of $P(\lambda) :=\left\|\Pi_M|\psi_l\rangle\right\|$ as a function of $\lambda =\left\|\Pi_M|\psi_0\rangle\right\|$,
    where we let $w=0.08$, $\delta=0.3$, and $l=\lceil \ln(2/\delta)/(2w) \rceil = 12$ (cf. Theorem~\ref{thm:main}).
    It shows that $P(\lambda) \geq \sqrt{1-\delta^2}\approx 0.95$ as long as $\lambda \geq w=0.08$.
    The explicit formula of $P(\lambda)$ is shown in Eq.~\eqref{eq:P_lambda}.}
    \label{fig:thm1_examp}
\end{figure}

\section{Proof of Theorem~\ref{thm:main}}\label{sec:proof_thm}

\subsection{The quasi-Chebyshev lemma}

To prove Theorem~\ref{thm:main}, we will need the following Lemma~\ref{lem:chebyshev} regarding the explicit formula of a recursive quasi-Chebyshev polynomial.

\begin{lemma}\label{lem:chebyshev}
    Suppose $\gamma \in (0,1]$.
    Consider the odd polynomial $a^\gamma_L(x)$ of degree $L=2l+1$ defined by the following recursive relation:
\begin{align}
a^\gamma_0(x) &= 1,\ a^\gamma_1(x)=x, \\
a^\gamma_{n+1}(x) &= x(1+e^{-i\phi_n}) a^\gamma_n(x) -e^{-i\phi_n} a^\gamma_{n-1}(x), \label{eq:a_L_recur}
\end{align}
    where the angles are 
\begin{equation}\label{eq:phi_n_def}
\phi_n = 2\arctan\left( \sqrt{1-\gamma^2} \tan(\frac{n}{L}\pi) \right), \ n=1\sim 2l.
\end{equation}
    Then the explicit formula of $a^\gamma_L(x)$ is
\begin{equation}
    a^\gamma_L(x) = T_L(x/\gamma) / T_L(1/\gamma),
\end{equation}
    where $T_L(x) \equiv \cos(L\arccos(x)) \equiv \cosh(L\,\mathrm{arccosh}(x))$ is the Chebyshev polynomial of the first kind.
\end{lemma}

The main contribution of this article is to provide an explicit and detailed proof of  Lemma~\ref{lem:chebyshev}.
As the argument is quite lengthy, the proof is deferred to Section~\ref{sec:prooflemma1}.
We now show that with Lemma~\ref{lem:chebyshev} in hand, one can prove Theorem~\ref{thm:main}.

\subsection{Relating failure amplitude to quasi-Chebyshev polynomials}

The core of relating Theorem~\ref{thm:main} to Lemma~\ref{lem:chebyshev} lies in the following Lemma~\ref{lem:phase_amp}, which may be of independent interest.
Specifically, the final failing amplitude $\left\| \Pi_{M}^\bot \ket{\psi_l} \right\|$ can be expressed as $\left| a_L(x) \right|$,
where $x = \left\| \Pi_{M}^\bot \ket{\psi_0} \right\|$ and $a_L(x)$ is a recursive quasi-Chebyshev polynomial,
and that the angles $\phi_n$ in the recursive relation of $a_L(x)$ are related to the angles $(\alpha_k,\beta_k)_{k=1}^l$ in a simple way:

\begin{lemma}\label{lem:phase_amp}
Let $x:= \left\| \Pi_{M}^\bot \ket{\psi_0} \right\|$ be
the length of the projection of the initial state $\ket{\psi_0}$ onto the subspace spanned by unmarked states.
Consider the final state $$\ket{\psi_l} := G(\alpha_l,\beta_l) \cdots G(\alpha_1,\beta_1) \ket{\psi_0},$$ where $G(\alpha,\beta)$ is defined by Eq.~\eqref{eq:generalized_iter}.
Then $\left\| \Pi_{M}^\bot \ket{\psi_l} \right\| = \left| a_L(x) \right|$, where the odd polynomial $a_L(x)$ of degree $L=2l+1$ is defined by the recursive relation:
$a_0(x)=1,a_1(x)=x$, and $a_{n+1}(x) = x(1+e^{-i\phi_n}) a_n(x) -e^{-i\phi_n} a_{n-1}(x)$ for $n=1\sim 2l$,
and the angles $\{\phi_n\}_{n=1}^{2l}$ are related to $(\alpha_k,\beta_k)_{k=1}^l$ by $\phi_{2k-1}=\pi-\alpha_k$ and $\phi_{2k} = \beta_k +\pi$ for $k=1\sim l$.
\end{lemma}

\begin{proof}
We will restrict ourselves to the invariant subspace $\mathcal{H}_0 = \mathrm{span} \{\ket{r}, \ket{t}\}$,
where $\ket{r} = \Pi_M^\bot \ket{\psi_0}/x$ and $\ket{t} = \Pi_M \ket{\psi_0}/\sqrt{1-x^2}$.
Then the initial state $\ket{\psi_0} = x\ket{r} +\sqrt{1-x^2}\ket{t}$ in $\mathcal{H}_0$, and thus we have $\ket{\psi_0} = R(x) \ket{r}$, where
\begin{equation}\label{eq:def_R_x}
    R(x) := 
    \left(\begin{array}{cc}
        x & \sqrt{1-x^2}\\
        \sqrt{1-x^2} & -x
    \end{array}\right).
\end{equation}
Recall from Eq.~\eqref{eq:generalized_iter} that
\begin{equation}
    G(\alpha,\beta) = \left( I - (1-e^{i\beta})\ket{\psi_0}\bra{\psi_0} \right)
    \cdot \left( I - (1-e^{i\alpha})\sum_{m\in M} \ket{m}\bra{m} \right).
\end{equation}
Thus we can write the expression of $G(\alpha,\beta)$ in $\mathcal{H}_0$ as
\begin{align}
    G(\alpha,\beta)
    &= R(x)  \left( I -(1-e^{i\beta}) \ket{r} \bra{r} \right) R(x)
    \cdot \left( I-(1-e^{i\alpha})\ket{t}\bra{t} \right) \\
    &= e^{i\beta} R(x) M_t(\beta) R(x) M_t(-\alpha),
\end{align}
where
\begin{equation}
    M_t(\phi) :=
    \left(\begin{array}{cc}
        1 & 0 \\ 0 & e^{-i\phi}
    \end{array}\right).
\end{equation}
Let
\begin{align}
    A(\phi) &:= R(x) M_t(\phi) \\
    &= \left(\begin{array}{cc}
x & e^{-i\phi} \sqrt{1-x^2} \\
\sqrt{1-x^2} & -x e^{-i\phi}
\end{array}\right). \label{eq:A_phi_matrix}
\end{align}
We can then write the expression of the final state $\ket{\psi_l}\in \mathcal{H}_0$ as
\begin{align}
    \ket{\psi_l} 
    &= G(\alpha_l,\beta_l) \cdots G(\alpha_1,\beta_1) \ket{\psi_0} \\
    &= \exp(i \sum_{n=1}^{l}\beta_n)
     A(\phi_{2l}')  A(\phi_{2l-1}') \cdots
     A(\phi_{2}')  A(\phi_{1}')  A(\phi_0') \ket{r}, \label{eq:psi_l_subspace}
\end{align}
where the angles $\phi_n'$ are defined by $\phi_0' =0$, $\phi_{2k-1}' =-\alpha_k$, and $\phi_{2k}' =\beta_k$, for $k=1\sim l$.

Consider $a_n(x)$ and $b_n(x)$ for $n=1\sim L$ defined by
\begin{equation}\label{eq:a_n_def}
    [a_n(x),\ \sqrt{1-x^2} b_n(x)]^T := A(\phi_{n-1}') \cdots A(\phi_0') \ket{r}.
\end{equation}
Combining Eq.~\eqref{eq:psi_l_subspace} and Eq.~\eqref{eq:a_n_def},
we know $\left| \braket{r|\psi_l}\right| = \left| a_L(x) \right|$.
Note that $\left\| \Pi_{M}^\bot \ket{\psi_l} \right\| = \left| \braket{r|\psi_l}\right|$ in $\mathcal{H}_0$.
Thus, it remains to show that $a_L(x)$ can also be defined by the desired recursive relation shown in Lemma~\ref{lem:phase_amp}.

From the matrix expression of $A(\phi)$ (Eq.~\eqref{eq:A_phi_matrix}) and Eq.~\eqref{eq:a_n_def}, we have the following coupled recursive relation of $a_n(x)$ and $b_n(x)$ for $n=0\sim 2l$:
\begin{align}
a_{n+1}(x) &= x a_n(x) +e^{-i\phi_n'}(1-x^2)b_n(x), \label{eq:recurr_1} \\
b_{n+1}(x) &= a_n(x) -x e^{-i\phi_n'} b_n(x), \label{eq:recurr_2}
\end{align}
where $a_0(x) := 1, b_0(x) := 0$.
To decouple this recursive relation,
we first obtain $b_{n}(x) = \left(a_{n-1}(x) - x a_{n}(x)\right)/(1-x^2)$ using ``(\ref{eq:recurr_1}) $+$ (\ref{eq:recurr_2}) $\times \frac{1-x^2}{x}$'' and by replacing $n$ with $n-1$,
and then plug $b_{n}(x)$ into Eq.~\eqref{eq:recurr_1}, resulting in:
\begin{equation}\label{eq:recurr_3}
    a_{n+1}(x) = x(1-e^{-i\phi_n'}) a_n(x) +e^{-i\phi_n'} a_{n-1}(x), 
\end{equation}
for $n=1\sim 2l$, and $a_0(x) = 1, a_1(x) = x$.
We now let $\phi_n := \phi_n' +\pi$.
Then $\phi_{2k-1}=\pi-\alpha_k$ and $\phi_{2k} = \beta_k +\pi$ for $k=1\sim l$,
and $a_{n+1}(x) = x(1+e^{-i\phi_n}) a_n(x) -e^{-i\phi_n} a_{n-1}(x)$ for $n=1\sim 2l$ from Eq.~\eqref{eq:recurr_3}, which is the desired result.
\end{proof}

\subsection{Finishing the proof of Theorem~\ref{thm:main}}

Using Lemma~\ref{lem:phase_amp} and Lemma~\ref{lem:chebyshev}, we can now prove Theorem~\ref{thm:main} as follows.

Recall that the parameters $(\alpha_k,\beta_k)_{k=1}^l$ in Theorem~\ref{thm:main} are:
\begin{equation}\label{eq:para_setting}
    \begin{cases}
        \alpha_k = 2\mathrm{arccot}\left( w \tan(\frac{2k-1}{2l+1}\pi) \right),\\
        \beta_k = -2\mathrm{arccot}\left( w \tan(\frac{2k}{2l+1}\pi) \right).
    \end{cases}
\end{equation}
Using Lemma~\ref{lem:phase_amp}, we know $\left\| \Pi_{M}^\bot \ket{\psi_l} \right\| = \left| a_L(x) \right|$ where $x:= \left\| \Pi_{M}^\bot \ket{\psi_0} \right\|$,
and the odd polynomial $a_L(x)$ of degree $L=2l+1$ is defined by the recursive relation:
$a_0(x)=1,a_1(x)=x$, and $a_{n+1}(x) = x(1+e^{-i\phi_n}) a_n(x) -e^{-i\phi_n} a_{n-1}(x)$ for $n=1\sim 2l$, and the angles $\{\phi_n\}_{n=1}^{2l}$ are:
\begin{equation}\label{eq:para_setting_2}
    \begin{cases}
        \phi_{2k-1}=\pi-\alpha_k = 2\arctan\left( w\tan(\frac{2k-1}{L}\pi) \right), \\
        \phi_{2k} = \beta_k +\pi = 2\arctan\left( w\tan(\frac{2k}{L}\pi) \right),
    \end{cases}
\end{equation}
for $k=1\sim l$, where we used Eq.~\eqref{eq:para_setting} in the second equality.

Let
\begin{equation}
    \gamma :=\sqrt{1-w^2},
\end{equation}
then $\phi_n = 2\arctan\left( \sqrt{1-\gamma^2} \tan(\frac{n}{L}\pi) \right)$ for $n=1\sim 2l$ by Eq.~\eqref{eq:para_setting_2}, which coincides with Eq.~\eqref{eq:phi_n_def}.
Thus, by Lemma~\ref{lem:chebyshev} we know $a_L(x) = T_L(x/\gamma) / T_L(1/\gamma)$.

In addition, using the fact that $\left\| \Pi_M \ket{\varphi} \right\| = \sqrt{1 -\left\| \Pi_M^\bot \ket{\varphi} \right\|^2}$ for any quantum state $\ket{\varphi}$,
we obtain the (earlier mentioned, cf. Fig.~\ref{fig:thm1_examp}) explicit formula of $P(\lambda)= \left\| \Pi_M \ket{\psi_l} \right\|$ about $\lambda = \left\|\Pi_M\ket{\psi_0}\right\|$ as:
\begin{equation}\label{eq:P_lambda}
    P(\lambda) = \sqrt{1 -\frac{T_L^2(\sqrt{1-\lambda^2}/\sqrt{1-w^2})}{T_L^2(1/\sqrt{1-w^2})}}.
\end{equation}

From the assumption $l \geq \ln(2/\delta)/(2w)$ in Theorem~\ref{thm:main}, we have $L \geq {\ln(2/\delta)}/{w}$.
We now show that it implies $ 1/T_L(1/\gamma) \leq \delta$.
Using the definition $T_L(x) \equiv \cosh(L\,\mathrm{arccosh}(x))$, we have
\begin{align}
&1/T_L(1/\gamma) \leq \delta \\
\Leftrightarrow\ & \cosh\left( L\, \mathrm{arccosh}(1/\gamma) \right)  \geq 1/\delta \\
\Leftrightarrow\ & L \geq \frac{\mathrm{arccosh}(1/\delta)}{\mathrm{arccosh}(1/\gamma)}. \label{eq:L_RHS}
\end{align}
Thus, it suffices to show that the RHS of Eq.~\eqref{eq:L_RHS} is not greater than ${\ln(2/\delta)}/{w}$.
From $\gamma =\sqrt{1-w^2}$ and the definition $\mathrm{arccosh}(x) = \ln(x+\sqrt{x^2-1})$, we have
\begin{align}
\frac{\mathrm{arccosh}(1/\delta)}{\mathrm{arccosh}(1/\gamma)} &= \frac{\ln\left( \frac{1}{\delta} +\sqrt{\frac{1}{\delta^2}-1} \right)}{\ln\left( \frac{1}{\sqrt{1-w^2}} +\sqrt{\frac{1}{1-w^2}-1} \right)} \\
&= \frac{\ln\big[ \left(1+\sqrt{1-\delta^2}\right)/\delta \big]}{\ln \big[ (1+w)/\sqrt{1-w^2} \big]} \\
&\leq \frac{\ln(2/\delta)}{\ln\sqrt{\frac{1+w}{1-w}}}
\leq \frac{\ln(2/\delta)}{w},
\end{align}
where the last line follows from $\ln\left( \frac{1+w}{1-w} \right) \geq 2w$,
which can be shown by taking derivative on both sides and observing that $\frac{2}{1-w^2} \geq 2$ holds for $w\in [0,1)$.

Finally, combining $\left\| \Pi_{M}^\bot \ket{\psi_l} \right\| = \left| a_L(x) \right|$,
and $a_L(x) = \frac{T_L(x/\gamma)}{T_L(1/\gamma)}$,
and $ 1/T_L(1/\gamma) \leq \delta$,
and the fact that $|T_L(x)|\leq 1$ as long as $|x|\leq 1$,
we have: $\left\| \Pi_{M}^\bot \ket{\psi_l} \right\| \leq \delta$ as long as $|x|\leq \gamma$.
Recall that $x= \left\| \Pi_{M}^\bot \ket{\psi_0} \right\|$ and $\gamma=\sqrt{1-w^2}$.
Therefore, $\| \Pi_M \ket{\psi_l} \| \geq \sqrt{1-\delta^2}$ as long as $\| \Pi_M \ket{\psi_0} \| \geq w$, which completes the proof of Theorem~\ref{thm:main}.

\section{Proof of Lemma~\ref{lem:chebyshev}}\label{sec:prooflemma1}

We copy Lemma~\ref{lem:chebyshev} in the following for convenience.
\begin{lemma}[Copy of Lemma~\ref{lem:chebyshev}]\label{lem:chebyshev_copy}
    Suppose $\gamma \in (0,1]$.
    Consider the odd polynomial $a^\gamma_L(x)$ of degree $L=2l+1$ defined by the following recursive relation:
\begin{align}
a^\gamma_0(x) &= 1,\ a^\gamma_1(x)=x, \\
a^\gamma_{n+1}(x) &= x(1+e^{-i\phi_n}) a^\gamma_n(x) -e^{-i\phi_n} a^\gamma_{n-1}(x), \label{eq:a_L_recur_copy}\\
\phi_n &= 2\arctan\left( \sqrt{1-\gamma^2} \tan(\frac{n}{L}\pi) \right), \ n=1\sim 2l. \label{eq:phi_n_def_copy}
\end{align}
    Then $a^\gamma_L(x) = T_L(x/\gamma) / T_L(1/\gamma)$,
    where $T_L(x) \equiv \cos(L\arccos(x)) \equiv \cosh(L\,\mathrm{arccosh}(x))$ is the Chebyshev polynomial of the first kind.
\end{lemma}

\subsection{Reducing to Eq.~\eqref{eq:final_N}}
To show $a_L^\gamma(x) = T_L(x/\gamma) / T_L(1/\gamma)$, we will separate $a_L^\gamma(x)$ into two parts to match the RHS: $a_L^\gamma(x) = N_L^\gamma(x) / D_L(\gamma)$.
Lemma~\ref{lem:D_L} below shows that $D_L(\gamma) = \gamma^L T_L(1/\gamma)$, and thus it suffices to show $N_L^\gamma(x) = \gamma^L T_L(x/\gamma)$ (Eq.~\eqref{eq:final_N}) to prove Lemma~\ref{lem:chebyshev_copy}.

Denote $t(n) := \tan(\frac{n}{L}\pi)$, $t_n := \sqrt{1-\gamma^2} \times t(n) :=\tan{\theta_n}$.
Then $\phi_n = 2\theta_n$ by Eq.~\eqref{eq:phi_n_def_copy}.
We expand $e^{-i\phi_n}$ into expressions regarding $t_n$ as follows:
\begin{align}
e^{-i\phi_n} &= e^{-i2\theta_n} = \cos(2\theta_n) -i\sin(2\theta_n) \\
&= \frac{\cos^2\theta_n - \sin^2\theta_n}{\cos^2\theta_n + \sin^2\theta_n} - i\frac{2\cos\theta_n \sin\theta_n}{\cos^2\theta_n + \sin^2\theta_n} \\
&= \frac{1-\tan^2\theta_n}{1+\tan^2\theta_n} -i\frac{2\tan\theta_n}{1+\tan^2\theta_n} \\
&= \frac{(1-it_n)^2}{(1-it_n)(1+it_n)} \\
&= \frac{1-it_n}{1+it_n}.
\end{align}
Therefore, the recursive relation of $a^\gamma_{n}(x)$ shown by Eq.~\eqref{eq:a_L_recur_copy} can be written as
\begin{equation}\label{eq:a_n_recur_t}
a^\gamma_{n+1}(x) = \frac{2x}{1+it_n} a^\gamma_n(x) - \frac{1-it_n}{1+it_n} a^\gamma_{n-1}(x).
\end{equation}

If we let $N^\gamma_n(x) := a^\gamma_n(x)  D_n(\gamma)$, where $D_n(\gamma) := \prod_{k=0}^{n-1}(1+it_k)$ for $n=1\sim L$ and $D_0(\gamma):=1$, then Eq.~\eqref{eq:a_n_recur_t} can be written as
\begin{equation}\label{eq:N_D_n_recur}
\frac{N^\gamma_{n+1}(x)}{D_{n+1}(\gamma)} = \frac{2x}{1+it_n} \frac{N^\gamma_n(x)}{D_n(\gamma)} - \frac{1-it_n}{1+it_n}  \frac{1+it_{n-1}}{1+it_{n-1}} \frac{N^\gamma_{n-1}(x)}{D_{n-1}(\gamma)},
\end{equation}
for $n=1\sim 2l$.
Multiplying both sides of Eq.~\eqref{eq:N_D_n_recur} by $D_{n+1}(\gamma)$, we obtain the recursive definition of $N^\gamma_L(x)$ as follows
\begin{align}
N^\gamma_{n+1}(x) &= 2x N^\gamma_{n}(x) -(1-it_n)(1+it_{n-1}) N^\gamma_{n-1}(x), \label{eq:N_recur}
\end{align}
for $n=1\sim 2l$, and $N^\gamma_0(x) = 1, N^\gamma_1(x) = x$.

From Lemma~\ref{lem:D_L} shown below, we know $D_L(\gamma) = \gamma^L T_L(1/\gamma)$.
We will later show in Section~\ref{subsec:N_T_equals} that $N^\gamma_L(x)$ has the following explicit formula:
\begin{equation}\label{eq:final_N}
N^\gamma_L(x) = \gamma^L T_L(x/\gamma).
\end{equation}
Thus $a^\gamma_n(x) =N^\gamma_L(x)/D_L(\gamma) = T_L(x/\gamma) / T_L(1/\gamma)$, completing the proof of Lemma~\ref{lem:chebyshev_copy}.

\begin{lemma}\label{lem:D_L}
Suppose $L=2l+1$ is odd.
Consider the degree-$2l$ polynomial $D_L(\gamma)$ defined by $D_L(\gamma) = \prod_{n=0}^{2l}(1+it_n)$,
where $t_n = \sqrt{1-\gamma^2} \times t(n)$ and $t(n) = \tan(\frac{n}{L}\pi)$.
Then we have:
\begin{equation}\label{eq:final_D}
    D_L(\gamma) = \gamma^L T_L(1/\gamma),
\end{equation}
where $T_L(x) = \cos(L\, \arccos x)$ is the Chebyshev polynomial of the first kind.
\end{lemma}

\begin{remark}
Eq.~\eqref{eq:final_D} is actually a special case of Eq.~\eqref{eq:final_N} with $x=1$.
Therefore, assuming Eq.~\eqref{eq:final_N} holds, we only need to show $D_L(\gamma) = N^\gamma_L(1)$ to prove Lemma~\ref{lem:D_L}.
The equality $D_L(\gamma) = N^\gamma_L(1)$ holds because (1) $D_n(\gamma) = \prod_{k=0}^{n-1}(1+it_k)$ satisfies the recursive relation of $N^\gamma_n(x)$ when $x=1$ (cf. Eq.~\eqref{eq:N_recur}), since $D_{n+1}(\gamma) = 2 D_n(\gamma) -(1-it_n)(1+it_{n-1}) D_{n-1}(\gamma)$, and (2) the initial terms also coincide, since $D_0(\gamma) = 1 = N^\gamma_0(1)$ and $D_1(\gamma)=1=N^\gamma_1(1)$.
\end{remark}

\begin{proof}[Proof of Lemma~\ref{lem:D_L}]
We now prove this lemma without assuming Eq.~\eqref{eq:final_N}.
Note that $t_0=0$ and $t_{L-n}=-t_n$, thus
\begin{align}
D_L(\gamma) &= \prod_{n=1}^{l} \left(1+i\sqrt{1-\gamma^2} t(n)\right) \left(1-i\sqrt{1-\gamma^2} t(n)\right) \\
&= \prod_{n=1}^{l} \left( 1+(1-\gamma^2) t(n)^2 \right). \label{eq:D_L_lin2}
\end{align}
By replacing $\gamma$ with $1/\gamma$ in Eq.~\eqref{eq:final_D}, it suffices to show the following equality:
\begin{equation}\label{eq:D_transform}
    \gamma^L D_L(1/\gamma) = T_L(\gamma),
\end{equation}
where the LHS of Eq.~\eqref{eq:D_transform} is equal to the following formula $p(\gamma)$ by Eq.~\eqref{eq:D_L_lin2}.
\begin{equation}
    p(\gamma) := \gamma \prod_{n=1}^{l}(\gamma^2+(\gamma^2-1) t(n)^2).
\end{equation}
Thus, it suffices to prove $p(\gamma) = T_L(\gamma)$.
Note that $p(1)  = 1 = T_L(1)$,
and the degree-$L$ odd polynomial $T_L(\gamma) = \cos(L\, \arccos \gamma)$ has $L$ zeros: $0$ and
\begin{align}
\pm \left\{\cos\left(\frac{n+1/2}{L}\pi\right)\right\}_{n=0}^{l-1}
&= \pm \left\{\sin\left(\frac{L/2 -(n+1/2)}{L}\pi\right)\right\}_{n=0}^{l-1} \\
&= \pm \left\{\sin\left(\frac{l-n}{L}\pi\right)\right\}_{n=0}^{l-1} \\
&= \pm \left\{\sin\left(\frac{n}{L}\pi\right)\right\}_{n=1}^{l},
\end{align}
which coincide with the zeros of $p(\gamma)$, since $\gamma = \pm \sin\left(\frac{n}{L}\pi\right)$ is the zeros of $\gamma^2+(\gamma^2-1) t(n)^2=0$.
As $p(\gamma)$ and $T_L(\gamma)$ are both degree-$L$ polynomials, $p(\gamma) = T_L(\gamma)$.
\end{proof}

\subsection{Proof of Eq.~\eqref{eq:final_N}}\label{subsec:N_T_equals}
In this subsection, we will compare the combinatorial interpretations of $2N^\gamma_L(x)$ and $2\gamma^L T_L(x/\gamma)$ to prove $2N^\gamma_L(x) =2\gamma^L T_L(x/\gamma)$, from which Eq.~\eqref{eq:final_N} holds.

\subsubsection{Combinatorial interpretation of $2N^\gamma_L(x)$}

Similar to Ref.~\cite[Theorem 4]{counting}, we will show that $2N^\gamma_L(x)$ (cf. Eq.~\eqref{eq:N_recur} for the recursive definition of $N^\gamma_L(x)$) can be regarded as counting weights of tilings on the $L$-star in Lemma~\ref{lem:combi_2Nx} below.
As a preliminary, we first introduce some combinatorial objects.

The `$L$-star' consists of $L$ positions $\braket{0,1,\cdots,L-1}$ modular $L$ equally distributed on a circle.
To form a `tiling' on the $L$-star, we need to cover all the $L$ positions using either `square' with weight $2x$ on any single position, or `domino' with weight $-(1-it_{n})(1+it_{n-1})$ on two consecutive positions $\braket{n,n-1}$.
Recall that $t_n = \sqrt{1-\gamma^2} \times t(n)$ and $t(n) = \tan(\frac{n}{L}\pi)$.
Thus $t_{n+L}=t_n$, which coincides with the fact that the positions on the $L$-star are modular $L$.
The weight of a tiling is the product of the weights of all its squares and dominos.
For example, a tiling on the $5$-star consisting of $1$ square and $2$ dominos is shown in Fig.~\ref{fig:5_star}.

\begin{figure}[htbp]
    \centering
    \includegraphics[width=0.6\linewidth]{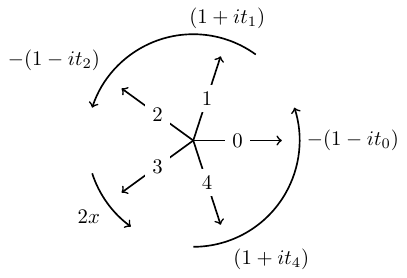}
    \caption{A tiling on the $5$-star with positions $\braket{0,1,2,3,4}$.
    It has one square with weight $2x$ on position $\braket{3}$, one domino with weight $-(1-it_2)(1+it_1)$ on positions $\braket{2,1}$, and one domino with weight $-(1-it_0)(1+it_4)$ on positions $\braket{0,4}$.
    The weight of this tiling is thus $2x(1-it_2)(1+it_1)(1-it_0)(1+it_4)$.}
    \label{fig:5_star}
\end{figure}

\begin{lemma}\label{lem:combi_2Nx}
$2N^\gamma_L(x)$ equals the sum of the weights of all the possible tilings on the $L$-star.
\end{lemma}

We first consider the combinatorial interpretation of $N^\gamma_L(x)$ and have the following Lemma~\ref{lem:combi_Nx}.

\begin{lemma}\label{lem:combi_Nx}
    $N^\gamma_L(x)$ equals the sum of the weights of all the possible `modified' tilings on the $L$-star.
    A `modified' tiling means that the dominos are not allowed to cover positions $\braket{0,L-1}$ and a square has weight $x$ instead of $2x$ when covering position $\braket{0}$.
\end{lemma}

\begin{proof}
In other words, we can cut the $L$-star between positions $\braket{0,L-1}$ to form a $L$-line with positions $\braket{L-1,\cdots,0}$.
Thus, it suffices to show the following claim.

\textbf{Claim:} $N^\gamma_L(x)$ counts the weights of all the possible tilings on a $L$-line.
Specifically, a tiling has all the positions covered by either square with weight $x$ at $\braket{0}$ or weight $2x$ at $\braket{n}$ for $n=1\sim(L-1)$,
or domino with weight $-(1-it_{n})(1+it_{n-1})$ at $\braket{n,n-1}$ for $n=1\sim(L-1)$.

Observe that any tiling on a $(n+1)$-line has either a square with weight $2x$ at $\braket{n}$, or a domino with weight $-(1-it_{n})(1+it_{n-1})$ at $\braket{n,n-1}$.
Thus, assuming the claim holds for $N^\gamma_n(x)$ and $N^\gamma_{n-1}(x)$,
the recursive relation $N^\gamma_{n+1}(x) =2x N^\gamma_{n}(x) -(1-it_n)(1+it_{n-1}) N^\gamma_{n-1}(x)$ implies that $N^\gamma_{n+1}(x)$ counts the weights of all the possible tilings on a $(n+1)$-line, so the claim also holds for $N^\gamma_{n+1}(x)$.
Finally, combing with $N^\gamma_0(x)=1$ and $N^\gamma_1(x)=x$, the claim holds by induction.
\end{proof}

We can now prove Lemma~\ref{lem:combi_2Nx} by multiplying the weight of each modified tiling in Lemma~\ref{lem:combi_Nx} by $2$.

\begin{proof}[Proof of Lemma~\ref{lem:combi_2Nx}]
From the combinatorial interpretation of $N^\gamma_L(x)$ as shown in Lemma~\ref{lem:combi_Nx},
we first divide all the possible modified tilings on the $L$-star into the following two types.
\begin{enumerate}
    \item[A.] Position $\braket{0}$ is covered by a square with weight $x$.
    \item[B.] Positions $\braket{0,1}$ are covered by a domino.
\end{enumerate}

If a modified tiling $T\in A$, then multiplying its weight by $2$ can be regarded as simply changing the weight of the square on position $\braket{0}$ from $x$ to $2x$.
Denote by $f(T)$ the obtained tiling on the $L$-star.

If a modified tiling $T\in B$, then $T$ is already a tiling on the $L$-star, and thus we cannot use the same technique as in $T\in A$.
Instead, we will add a new tiling $g(T)$ on the $L$-star, where $g$ is a reflection across the horizontal line passing position $\braket{0}$.
Specifically, $g$ reflects all the consisting squares and dominos of $T$ as follows:
it moves the square on position $\braket{n\neq 0}$ with weight $2x$ to position $\braket{L-n}$,
and moves the domino on positions $\braket{n,n-1}$ with weight $-(1-it_{n})(1+it_{n-1}) = -(1+it_{L-n})(1-it_{L-n+1})$ to positions $\braket{L-n,L-n+1}$,
where we've used $t_{L-n}=-t_n$.
Note that the $g(T)$ is a valid tiling on the $L$-star,
and that $g(T)$ has the same weight as $T$.
Furthermore, reflecting twice maps $T$ back to itself, thus $g^2=I$.

Denote by $C$ the set of all possible tilings on the $L$-star.
The above process shows that $f(A) \cup B \cup g(B) \subseteq C$.
The other direction of inclusion can be seen by the fact that any tiling $T\in C$ belongs to one and only one of the following three cases:
\begin{enumerate}
    \item Position $\braket{0}$ of $T$ is covered by a square with weight $2x$, then we can simply change the square's weight to $x$ and obtain a modified tiling $T'\in A$. Thus $T = f(T') \in f(A)$.
    \item Positions $\braket{1,0}$ of $T$ are covered by a domino, then $T\in B$.
    \item Positions $\braket{0,L-1}$ of $T$ are covered by a domino, then $g(T) \in B$.
    By $g^2 = I$ we have $T\in g(B)$.
\end{enumerate}
Therefore, $f(A) \cup B \cup g(B) = C$.
It remains to show that the  tilings in $f(A)$, $B$, and $g(B)$ are distinct.
First, these three sets do not intersect with each other because they consist of three different types of tilings as shown above.
Second, tilings in $f(A)$ and $g(B)$ are all distinct, since the invertible maps $f$ and $g$ are bijections and the modified tilings in $A$ and $B$ are distinct.
\end{proof}

\subsubsection{Combinatorial interpretation of $2\gamma^L T_L(x/\gamma)$}

\begin{lemma}\label{lem:combi_Tgamma}

$2\gamma^L T_L(x/\gamma)$ has almost the same combinatorial interpretation of $2N^\gamma_L(x)$ as shown in Lemma~\ref{lem:combi_2Nx}, and the only difference is that the weight of dominos change from $-(1-it_{n})(1+it_{n-1})$ to $-\gamma^2$.

\end{lemma}

\begin{proof}
The Chebyshev polynomial of the first kind $T_n(x) =\cos(n \arccos x)$ has the following recursive relation:
\begin{align}
    T_0(x) &=1, T_1(x) =x, \\
    T_{n+1}(x) &= 2x T_n(x) -T_{n-1}(x) \label{eq:T_recurr}.
\end{align}
Comparing Eq.~\eqref{eq:T_recurr} with the recursive definition of $N^\gamma_L(x)$ in Eq.~\eqref{eq:N_recur},
and recalling that $t_n = \sqrt{1-\gamma^2} \times t(n)$,
we know $T_L(x)$ is a special case of $N^\gamma_L(x)$ with $\gamma=1$.
Thus $2T_L(x)$ has almost the same combinatorial interpretation of $2N^\gamma_L(x)$ as shown in Lemma~\ref{lem:combi_2Nx},
except that the weight of dominos change from $-(1-it_{n})(1+it_{n-1})$ to $-1$.
Then, by multiplying each of the $L$ positions in every tiling by $\gamma$ and substituting $x$ with $x/\gamma$,
we obtain the desired combinatorial interpretation of $2\gamma^L T_L(x/\gamma)$.
\end{proof}

\subsubsection{Comparing the coefficients}

To prove $2N^\gamma_L(x) =2\gamma^L T_L(x/\gamma)$, it suffices to show that the coefficients of $x^{n_s}$ where $n_s \in\{L,L-2,\cdots,1\}$ are the same for polynomials $2N^\gamma_L(x)$ and $2\gamma^L T_L(x/\gamma)$.
Comparing their combinatorial interpretations as shown by Lemma~\ref{lem:combi_2Nx} and Lemma~\ref{lem:combi_Tgamma},
we let $w =\sqrt{1-\gamma^2}$, and thus the weight of domino on positions $\braket{n,n-1}$ in $2N^\gamma_L(x)$ becomes $-(1-it(n)w)(1+it(n-1)w)$, where $t(n) = \tan(\frac{n}{L}\pi)$.
And the weight of domino in $2\gamma^L T_L(x/\gamma)$ becomes $-\gamma^2 =-(1-w)(1+w)$.
The weight of squares in both interpretations remains $2x$.
Comparing the total weights of tilings with $n_s$ squares (contributing to the coefficient of $x^{n_s}$), it suffices to show the following Lemma~\ref{lem:compare_w}.

\begin{lemma}\label{lem:compare_w}
The following two types of tilings on the $L$-star have the same total weights.
In both types of tilings, the weight of any square is $1$, and the number of dominos is $n_d =(L-n_s)/2 \in \{0,1,\cdots,l\}$, where $n_s \in\{L,L-2,\cdots,1\}$.
\begin{enumerate}
    \item[A.] The domino has weight $(1-it(n)w)(1+it(n-1)w)$ on positions $\braket{n,n-1}$, where $t(n) = \tan(\frac{n}{L}\pi)$.
    \item[B.] The domino has weight $(1-w)(1+w)$.
\end{enumerate}
\end{lemma}

\begin{proof}
To show the total weights of tilings of type A and type B are the same, we will compare the coefficient of $w^k$ for $k\in\{1,2,\cdots, L\}$.
The coefficients of $w^0$, or the constant terms, are the same for both types and are both equal to the number of all possible tilings on the $L$-star with $n_d$ dominos.

The tilings of type A or type B can be categorized into groups of size $L$, where tilings in each group have the same distribution of dominos (or squares) up to rotations.
Specifically, for a tiling $T\in A$, denote by $T(j)\in A$ the new tiling obtained from shifting all the squares and dominos of $T$ by $j\in[L]:=\{0,1,\cdots,L-1\}$ positions,
and updating the weights of shifted dominos correspondingly so that $T(j)\in A$.
The group to which tiling $T\in A$ belongs is $\{T(j):j\in[L]\} \subseteq A$, since $T(0)=T$.
See Fig.~\ref{fig:5_terms} for an example of a $5$-sized group $\{T(j)\}_{j=0}^{4}$ of type A tilings on the $5$-star with $n_d=2$ dominos.

Denote by $T'\in B$ the counterpart of $T\in A$ obtained by changing the weights of dominos of $T$ from $(1-it(n)w)(1+it(n-1)w)$ to $(1-w)(1+w)$,
and $T'(j)\in B$ the tiling obtained by rotating $T'$ by $j$ positions.
Then, the group $\{T'(j):j\in[L]\} \subseteq B$ is a counterpart of $\{T(j):j\in[L]\} \subseteq A$.
Since each tiling $T$ belongs to one $L$-sized group ($T=T(0)\in\{T(j):j\in[L]\}$), and two $L$-sized group do not intersect (if two groups intersect, i.e. $T\in G_1\cap G_2$, then the two groups are the same: $G_1 =\{T(j):j\in[L]\} =G_2$), it suffices to prove that the coefficient of $w^k$ for $k\in\{1,2,\cdots, L\}$ in the total weights of each $L$-sized group $\{T(j):j\in[L]\} \subseteq A$ is the same as its counterpart group $\{T'(j):j\in[L]\} \subseteq B$.

The coefficient of $w^k$ in the total weights of a $L$-sized group $\{T(j):j\in[L]\} \subseteq A$ can be divided into different $L$-sized sets of terms ($L$-terms for short).
Specifically, if the coefficient of $w^k$ contains the term $\prod_{m=1}^{k} (-1)^{d_m} it(l_m)$ when calculating the weight of $T\in \{T(j):j\in[L]\}$,
where $\{l_m\}_{m=1}^k \subseteq [L]$ and $d_m\in\{0,1\}$ depends on the distribution of dominos in $T$,
then the coefficient of $w^k$ will contain $\prod_{m=1}^{k} (-1)^{d_m} it(l_m+j)$ for $j\in[L]$, by the definition of $T(j)$.
See Fig.~\ref{fig:5_terms} for an example of a $5$-terms $\{(-i)t(0+j)\cdot it(2+j)\}_{j=0}^{4}$ appearing in the coefficient of $w^2$ in the total weights of a $5$-sized group $\{T(j)\}_{j=0}^{4}$.

We will later prove in Lemma~\ref{lem:sum_prod}
that each $L$-terms sum up to either $L$ or zero depending on whether $k$ is even or odd: 
\begin{equation}\label{eq:sum_prod_1}
\sum_{j\in [L]} \prod_{m=1}^{k} it(l_m+j) =\delta(2\vert k)L,
\end{equation}
where $\delta(2\vert k) = 1$ if $k$ is even, and $\delta(2\vert k) = 0$ if $k$ is odd.

When $k$ is even, note that when $\prod_{m=1}^{k} (-1)^{d_m} it(l_m+j)$ appears in the coefficient of $w^k$ in $T(j)$,
the corresponding term $\prod_{m=1}^{k} (-1)^{d_m}$ appears in the coefficient of $w^k$ in the counterpart $T'(j)$.
Equation (\ref{eq:sum_prod_1}) then implies $\sum_{j\in[L]} \prod_{m=1}^{k} (-1)^{d_m} it(l_m+j) =L\prod_{m=1}^{k} (-1)^{d_m}$,
saying that the sum of these $L$-terms in $\{T(j):j\in[L]\}$ is the same as the sum of corresponding $L$-terms in $\{T'(j):j\in[L]\}$.
Since the coefficient of $w^k$ in $\{T(j):j\in[L]\}$ or $\{T'(j):j\in[L]\}$ consists of different such $L$-terms,
the coefficient of $w^k$ in $\{T(j):j\in[L]\} \subseteq A$ is the same as that in $\{T'(j):j\in[L]\} \subseteq B$.

When $k$ is odd, Eq.~\eqref{eq:sum_prod_1} implies that the coefficient of $w^{k}$ in any $L$-sized group $\{T(j):j\in[L]\} \subseteq A$ is zero,
since the coefficient of $w^k$ in $\{T(j):j\in[L]\}$ consists of different such $L$-terms that sum up to zero.
The coefficient of $w^{k}$ in the total weights of tilings in the counterpart group $\{T'(j):j\in[L]\} \subseteq B$ is also zero,
because the weight of any tiling in the counterpart group $\{T'(j):j\in[L]\} \subseteq B$ does not contain odd powers of $w$,
which follows from the fact that its dominos have weight $(1-w^2)$.
\end{proof}

\begin{figure*}[htbp]
    \centering    \includegraphics[width=1\linewidth]{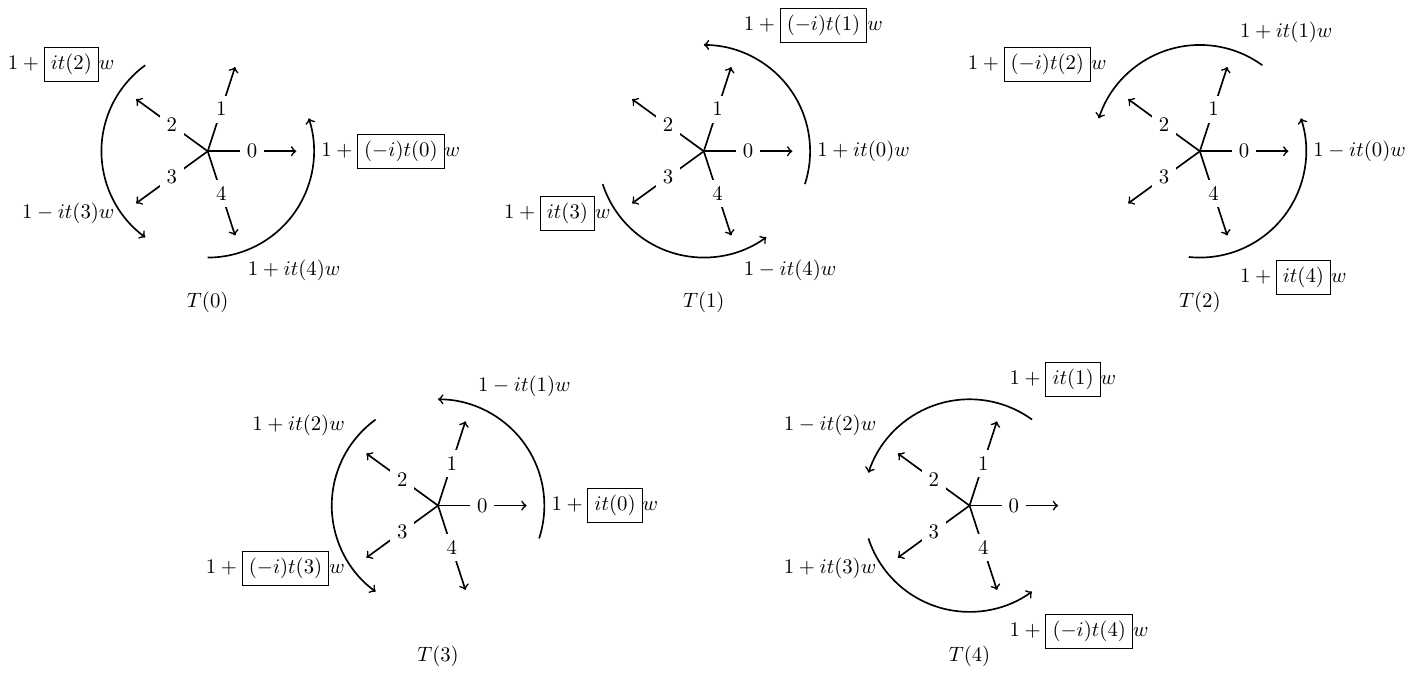}
    \caption{A group of type A (see Lemma~\ref{lem:compare_w} for its definition) tilings $\{T(j)\}_{j=0}^{4}$ on the $5$-star with $n_d=2$ dominos, where $T(j)$ is obtained from shifting all the dominos in $T(0)$ by $j$ positions and updating the weights of shifted dominos correspondingly. The set of products of the two boxed coefficients in each tiling, i.e. $\{(-i)t(0+j)\cdot it(2+j)\}_{j=0}^{4}$, is a $5$-sized set of terms ($5$-terms) appearing in the coefficients of $w^2$ in the total weights of the group $\{T(j)\}_{j=0}^{4}$, and this $L$-terms sum up to $5$ by Eq.~\eqref{eq:sum_prod_1}.}
    \label{fig:5_terms}
\end{figure*}


The following lemma regards a specific sum of products of tangents, which may be of independent interest.

\begin{lemma}\label{lem:sum_prod}
Assume the integer $L\geq 3$ is odd.
Denote $t(n) := \tan(n\pi/L)$.
Note that $t(n) = t(n\,\mathrm{mod}\,L)$.
For $k\in\{1,\cdots,L\}$, consider any $k$-sized subset $\{l_m\}_{m=1}^k \subseteq [L] :=\{0,1,\cdots,L-1\}$, then
\begin{equation}\label{eq:sum_prod_2}
\sum_{j\in [L]} \prod_{m=1}^{k} it(l_m+j) = \delta(2\vert k)L,
\end{equation}
where $\delta(2\vert k) \in\{0,1\}$ indicates whether $k$ is odd or even.
\end{lemma}

\begin{proof}
We will prove Lemma~\ref{lem:sum_prod} by induction.

\textbf{Base case I.} When $k=L$, we have $\prod_{m=1}^{L} t(l_m +j)=0$,
since $\{(l_m+j)\,\mathrm{mod}\,L\}_{m=1}^{L} = [L]$ for any $j\in[L]$ and $t(0)=0$.
Thus, Eq.~\eqref{eq:sum_prod_2} holds since $L$ is odd.

\textbf{Base case II.} When $k=L-1$, suppose $\{l_m\}_{m=1}^{L-1} = [L]\setminus \{l_0\}$, then $\{(l_m +j)\,\mathrm{mod}\,L\}_{m=1}^{L} = [L]\setminus \{(j+l_0)\,\mathrm{mod}\,L\}$.
Note that $\{(j+l_0)\,\mathrm{mod}\,L\}_{j\in[L]} =[L]$, thus Eq.~\eqref{eq:sum_prod_2} becomes:
\begin{equation}
    \sum_{\{d_m \} \subset [L] \atop |\{d_m\}| = L-1} \prod_{m = 1}^{k} it(d_m) = \delta(2\vert k)L,
\end{equation}
which is a special case of the following equality with $k=L-1$.
\begin{equation}\label{eq:vieta_dm}
\sum_{\{d_m \} \subset [L] \atop |\{d_m\}| = k} \prod_{m = 1}^{k} it(d_m) = \delta(2\vert k)
\left(\begin{array}{@{}c@{}}
L \\
k
\end{array}\right)
\end{equation}

We now prove Eq.~\eqref{eq:vieta_dm}.
Using Euler's formula $e^{ix} = \cos{x} +i\sin{x}$,
we have $e^{iLx} = (e^{ix})^L = \big( 1 +i\tan{x} \big)^L / \cos^L{x}$.
Expanding the latter binomial, we have:
\begin{equation}\label{eq:e_Lx}
\cos^L(x) e^{iLx} = \sum_{k=0}^{L}
\left(\begin{array}{@{}c@{}}
L \\
k
\end{array}\right)
(i\tan x)^k
\end{equation}
When $x\in\{\frac{n}{L}\pi\}_{n=0}^{L-1}$,
the imaginary part of $e^{iLx}=e^{in\pi}$ is zero,
and the imaginary part of $(i\tan(x))^k$ can be written as $\delta(2\vert k)(it(n))^k$.
Thus, Eq.~\eqref{eq:e_Lx} implies:
\begin{equation}
0 = \sum_{k=0}^{L} \delta(2\vert k)
\left(\begin{array}{@{}c@{}}
L \\
k
\end{array}\right)
(it(n))^k.
\end{equation}
Therefore, the polynomial $\sum_{k=0}^{L} \delta(2\vert k)
\left(\begin{array}{@{}c@{}}
L \\
k
\end{array}\right)
x^k$ has $L$ different zeros $\{it(n)\}_{n=0}^{L-1}$.
Using Vieta's Theorem or zero-coefficient relationship, we obtain Eq.~\eqref{eq:vieta_dm}.

\textbf{Induction step.} Assume Eq.~\eqref{eq:sum_prod_2} holds for $k+2,k+1$,
we now show that Eq.~\eqref{eq:sum_prod_2} also holds for $k \leq L-2$.

For any subset $\{l_m\}_{m=1}^k \subseteq [L]$ of size $k$,
since $k\leq L-2$,
we can find two different $l_{k+1},l_{k+2} \in [L]\setminus \{l_m\}_{m=1}^k$ such that $\{l_m\}_{m=1}^{k+2} \subseteq [L]$.
Since subsets $\{l_m\}_{m=1}^{k} \cup \{l_{k+1}\}$ and $\{l_m\}_{m=1}^{k} \cup \{l_{k+2}\}$ are both of size $k+1$,
using the induction of hypothesis for $k+2$, we have:
\begin{align}
0 &= \sum_{j\in[L]} it(l_{k+2}+j) \prod_{m=1}^{k} it(l_m+j)
\nonumber\\&\quad
-\sum_{j\in[L]} it(l_{k+1}+j) \prod_{m=1}^{k} it(l_m+j) \\
&= it(l_{k+2}-l_{k+1}) \sum_{j\in[L]}\prod_{m=1}^{k} it(l_m+j)
\nonumber\\&\quad
-it(l_{k+2}-l_{k+1}) \sum_{j\in[L]}
\prod_{m=1}^{k+2} it(l_m+j), \label{eq:induction_line3}
\end{align}
where we used the trigonometric identity $\tan(x)-\tan(y) =\tan(x-y)\left( 1 -i\tan(x) i\tan(y) \right)$ in Eq.~\eqref{eq:induction_line3}.
Since $l_{k+2} \neq l_{k+1}$,
the fact that Eq.~\eqref{eq:induction_line3} equals zero implies:
\begin{align}
\sum_{j\in[L]} \prod_{m=1}^{k} it(l_m+j) &= \sum_{j\in[L]} \prod_{m=1}^{k+2} it(l_m+j)\\
&= \delta(2\vert k+2)L \label{eq:kplus2_line2}\\
&= \delta(2\vert k)L \label{eq:kplus2_line3},
\end{align}
where we used the induction of hypothesis for $k+1$ in Eq.~\eqref{eq:kplus2_line2},
and the fact that $k+2$ has the same parity as $k$ in Eq.~\eqref{eq:kplus2_line3}.
Thus, Eq.~\eqref{eq:sum_prod_2} also holds for $k$.
\end{proof}

\section{Conclusions}
In this article, we have reviewed the fixed-point quantum search that overcomes the souffle problem while maintaining the quadratic speedup, as it always finds a marked state with high probability as long as the proportion of marked states is greater than a predetermined lower bound.
The closed-form angle parameters in the fixed-point quantum search rely on a lemma regarding the explicit formula of recursive quasi-Chebyshev polynomials, but its proof is not given explicitly in the original paper.
In this work, we have provided a detailed proof of this quasi-Chebyshev lemma, thus providing a sound foundation for the correctness of the fixed-point quantum search.

To prove the quasi-Chebyshev lemma, we have used nontrivial techniques and tools such as combinatorial interpretations of quasi-Chebyshev polynomials as counting weights of tilings on the $L$-star, Euler's formula, trigonometric identities of the tangent function, binomial Theorem, Vieta's Theorem, and mathematical induction.
It's natural to wonder how the authors found the closed-form angle parameters in the first place.
The lemma may be of independent interest, as it has been a key component in overcoming the souffle problem of quantum walk search on complete bipartite graphs, and it is also central to the quantum phase discrimination algorithm with applications to quantum search on graphs.
It will be interesting to find more applications of the quasi-Chebyshev lemma.

\section*{Acknowledgement}
This work was supported by the National Key Research and Development Program of China (Grant No.2024YFB4504004), the National Natural Science Foundation of China (Grant No. 92465202, 62272492, 12447107),  the Guangdong Provincial Quantum Science Strategic Initiative (Grant No. GDZX2303007, GDZX2403001), the Guangzhou Science and Technology Program (Grant No. 2024A04J4892).

\section*{Competing Interest}
The authors declare that they have no competing interests or financial conflicts to disclose.






\bibliographystyle{fcs}
\bibliography{ref}

\begin{biography}{photo_LGZ}
Guanzhong Li received his BS degree in mathematics from Sun Yat-sen University, China in 2021. He is currently pursuing his PhD degree in computer science at Sun Yat-sen University, China. His research interests include quantum computing, quantum algorithms, quantum walks.
\end{biography}

\begin{biography}{photo_FSG}
    Shiguang Feng received the B.S. degree in computer science and technology from Shandong Agricultural University, Tai'an, China, in 2006; the Ph.D. degree in logic from Sun Yat-sen University, Guangzhou, China, in 2012; and the Doctor of Natural Science degree in computer science from Leipzig University, Leipzig, Germany, in 2016. He is currently an associate researcher with the School of Computer Science and Engineering, Sun Yat-sen University, Guangzhou, China. His current research interests include reversible logic synthesis, quantum algorithms and mathematical logic.
\end{biography}

\begin{biography}{photo_LLZ}
Lvzhou Li received his PhD degree in Computer Science from Sun Yat-sen University, China in 2009 and then worked in Sun Yat-sen University, China. Now he is a professor of the School of Computer Science and Engineering, Sun Yat-sen University, China. His research interests are quantum algorithm,  quantum circuit synthesis and optimization, and quantum machine learning.
\end{biography}

\end{document}